\documentclass[submission,copyright,creativecommons]{eptcs}
\providecommand{\event}{GandALF 2017} 
\usepackage{breakurl}             

\title{MK-fuzzy Automata and MSO Logics\thanks{Supported by the Austrian Research Promotion Agency (FFG) in the frame
of the BRIDGE program 846003 \enquote{LogicGuard II}.} }
\author{Manfred Droste
\institute{Institut f\"{u}r Informatik \\ Universit\"{a}t Leipzig \\D-04109 Leipzig, Germany}
\email{droste@informatik.uni-leipzig.de}
\and
Temur Kutsia
\institute{Research Institute for Symbolic Computation (RISC)\\Johannes Kepler University \\A-4040 Linz, Austria \\}
\email{Temur.Kutsia@risc.jku.at}
\and
George Rahonis
\institute{Department of Mathematics \\Aristotle University of Thessaloniki\\54124 Thessaloniki, Greece}
\email{grahonis@math.auth.gr}
\and 
Wolfgang Schreiner
\institute{Research Institute for Symbolic Computation (RISC)\\Johannes Kepler University \\A-4040 Linz, Austria \\}
\email{Wolfgang.Schreiner@risc.jku.at}
}
\def\titlerunning{MK-fuzzy Automata and MSO Logics}
\def\authorrunning{M. Droste, T. Kutsia, G. Rahonis \& W. Schreiner}

\usepackage{amssymb}
\usepackage{amsfonts}
\usepackage{amsmath}
\usepackage{csquotes}
\usepackage{stmaryrd}
\allowdisplaybreaks

\newtheorem{theorem}{Theorem}

\newtheorem{definition}[theorem]{Definition}
\newtheorem{example}[theorem]{Example}
\newtheorem{lemma}[theorem]{Lemma}
\newtheorem{corollary}[theorem]{Corollary}
\newtheorem{remark}[theorem]{Remark}

\newtheorem{proposition}[theorem]{Proposition}

\newenvironment{proof}{\textbf{Proof.}}

\hypersetup{
  colorlinks   = true, 
  urlcolor     = black, 
  linkcolor    = black, 
  citecolor    = black 
}

\def\ter{\mathit{ter}}

\begin{document}

\maketitle

\begin{abstract}
We introduce MK-fuzzy automata over a bimonoid $K$ which is related to the fuzzification of the McCarthy-Kleene logic. Our automata are inspired by, and intend to contribute to, practical applications being in development  in a project on runtime network monitoring based on predicate logic. We investigate closure properties of the class of recognizable MK-fuzzy languages accepted by MK-fuzzy automata as well as of deterministically recognizable MK-fuzzy languages accepted by their deterministic counterparts. Moreover, we establish a Nivat-like result for recognizable MK-fuzzy languages. We introduce an MK-fuzzy MSO logic and show the expressive equivalence of a fragment of this logic with MK-fuzzy automata, i.e., a B\"{u}chi type theorem.  
\end{abstract}

\section{Introduction}
Fuzzy automata constitute a special model of weighted automata but historically have been defined and studied separately, mostly inspired by fuzzy logic theory. The original fuzzy automaton model assigned to words values from the lattice $[0,1]$ with the usual $\max$ and $\min$ operations. Later on, fuzzy automata were investigated also over more general structures like for instance lattices, residuated lattices, and $l$-monoids. Several real world applications are modelled by fuzzy automata. We refer the reader to \cite{Mo:Fu} for fuzzy automata theory and applications, to \cite{Ra:Fu} for a generalization of them and their connection to weighted automata, and to \cite{Ah:Fu} for fuzzy semirings related to automata. For weighted automata theory, the interested reader should consult for instance \cite{Dr:Han,Dr:Au,Dr:Co}. 

On the other hand, McCarthy-Kleene logic (MK-logic for short), a combination of three-valued logics of  McCarthy~\cite{Mc:Ba} and Kleene~\cite{Kl:In}, has been introduced in \cite{Ko:Fo,Av:Pr} to reason about computation errors. The original idea, according to~\cite{Av:Pr}, was to distinguish between two types of errors: critical ones, which make the whole computation stop and cause a total failure of the program, and non-critical ones, which stop only part of the computation and can be fixed or circumvented by a success in some other part. MK-logic is a four-valued logic, where alongside the truth values $t$ (true) and $f$ (false) there are also $u$ (undefined, which originates from Kleene's logic) and $e$ (error, which comes from McCarthy's logic). In this combination, `undefined' is intended to represent non-critical errors, while `error' is reserved for critical ones. As in McCarthy's logic, interpretation of binary connectives is asymmetric, which means, for instance, that the disjunction of $t$ and $e$ is $t$, while the disjunction of $e$ and $t$ gives $e$.  In the combination it is assumed that $e$ prevails $u$ in whatever order they appear. 

MK-logic has found an application in the LogicGuard project~\cite{Lo:G1,Lo:G2,Ku:Lo,DBLP:conf/rv/CernaSK16} which pursues research on network security, developing a specification and verification formalism and tool for runtime network monitoring based on predicate logic. A monitor, which is a logical formula (usually with quantifiers), is interpreted over a network (an infinite stream of messages). The goal is to check whether the property specified in the monitor is satisfied by the stream, and report violating messages, if any. For instance, the following monitoring formula
\[ \text{\textbf{monitor} } x : p(x) \Rightarrow \text{\textbf{exists} } y \text{ \textbf{with} } x \le y \le x + T : q(x, y)\]
investigates for every stream position $x$ that satisfies $p(x)$ whether there exists some position y in range $[x, x + T]$
such that property $q(x, y)$ holds. 
Operationally, the monitor formula is translated into a program, which accepts stream messages one after the other, keeps evaluating the monitored property on the known part of the stream, and if it is violated (i.e., its truth value becomes $f$), reports the message that caused the violation. At each moment, the monitor observes only a finite initial part of the stream. Hence, it is not always possible to decide whether the property holds or not (`not enough' messages have arrived). In this case, a new copy of the current instance of the monitoring formula is created. Its truth value is $u$: undefinedness here really corresponds to `unknown', not to a non-critical error. The copy is added to the pile of copies of some previous instances, which also wait to be decided. Each of these copies will be evaluated for the incoming messages and will be removed from consideration if its truth value becomes $t$ or $f$. In the latter case, the violated message is reported. If something causes an error (i.e., if the truth value $e$ is generated for some reason), monitoring stops.
The LogicGuard framework has met the expectations of the developers, being successfully used for runtime network monitoring. As the next step, it is planned to deploy it for new application scenarios such as, for instance, ``Internet of Things''. Such applications pose new challenges, related to the difficulties with quantification of decisions, or to the fact that it is not a priori clear what the expectations of a correct execution of a system are. To deal with such problems, reasoning with some kind of probabilistic or fuzzy knowledge is required. As the first step towards this direction, we envisage the extension of the LogicGuard specification language to a fuzzy quantified logic that is able to handle specifications including uncertainty and vagueness. On this strand, and for the development of the fuzzification of the MK-logic and relative models, we introduce MK-fuzzy automata, and this paper is a first attempt to study these models. 
Our MK-fuzzy automata assign, to words, values from the bimonoid 
$$K=\{(t,f,u,e) \in [0,1]^4 \mid t+f+u+e = 1 \}$$
where its operations, called MK-disjunction and MK-conjunction, are inspired by the fuzzification of the MK-logic. Formal series with values in $K$ are called MK-fuzzy languages. 

Classical operations in formal series over semirings cannot be defined in the usual way over bimonoids due to the lack of commutativity and distributivity properties. Notable examples are the Cauchy product and the star operation. If the weight structure is weaker than a semiring, for instance a bimonoid like in our case, then the lack of commutativity, distributivity, and multiplicative zero properties has a serious impact on the automata models considered over such a weight structure. For instance the value assigned by the automaton to a word cannot be defined in the usual way. Due to these difficulties, and since no interesting bimonoid structures have been considered so far, there is a lack of work on weighted automata over bimonoids. According to our best knowledge, the most relative works deal with automata and transducers over strong bimonoids where the first operation is commutative and there is a multiplicative zero \cite{Ci:De,Dr:St,Li:Th}. 
For our MK-fuzzy automata, where a multiplicative zero is missing from the bimonoid $K$, we consider a set of initial states, a set of transitions, and a set of final states and define on these sets the initial distribution, the mapping assigning truth values to the transitions of the automaton, and the terminal distribution, respectively. Our model is nondeterministic. Since the MK-disjunction is not commutative, we require the state set of the MK-fuzzy automaton to be linearly ordered. Then the paths of the automaton over any word, can be ordered according to lexicographic order, and hence we can define the value of $K$ assigned by the MK-fuzzy automaton to the given word.  

We show that the class of recognizable MK-fuzzy languages accepted by MK-fuzzy automata is closed under MK-disjunction, strict alphabetic homomorphisms and inverse strict alphabetic homomorphisms. Moreover, we establish a Nivat-like decomposition result showing that recognizable MK-fuzzy languages can be obtained from very particular MK-fuzzy automata (in fact, with only one state), restriction to recognizable languages and strict alphabetic homomorphisms. 
We introduce also the deterministic counterpart of our model and show that the class of MK-fuzzy languages accepted by these automata, called deterministically recognizable, is closed under MK-disjunction with scalars. The Cauchy product of two deterministically recognizable MK-fuzzy languages is a recognizable MK-fuzzy language. Due to the structure of the bimonoid $K$, we can define several notions of supports of MK-fuzzy languages. We show that the strong support, related to the first component of the elements in $K$, of a deterministically recognizable MK-fuzzy language is a recognizable language. 
Furthermore, we introduce an MK-fuzzy MSO logic and determine a fragment of sentences which is expressively equivalent to the class of MK-fuzzy automata, i.e., a B\"uchi type theorem.  

\section{Preliminaries}

Let $A$ be an alphabet, i.e., a finite nonempty set. As usually, we denote by
$A^{\ast}$ the set of all finite words over $A$ and define $A^{+}=A^{\ast}%
\setminus \{\varepsilon\}$, where $\varepsilon$ is the empty word. The length of a word $w$, i.e., the number of the letters of $w$ is denoted as usual by $|w|$. A word $w=a_0\ldots a_{n-1}$ over $A$, with $a_0, \ldots, a_{n-1} \in A$, is written also as $w=w(0)\ldots w(n-1)$ with $w(i)=a_i$ for every $0 \leq i \leq n-1$. Assume now that $\leq$ is a linear order on $A$. The lexicographic order $\leq_{lex}$ on $A^*$ is defined as follows:
$$w \leq_{lex} u \quad \text{ iff } \quad ((u=wv \text{ with }v \in A^*) \text{ or } (w=vav', \ u=vbv'', \ v \in A^*, \ a,b \in A \text{ with } a<b))$$
for every $w,u \in A^*$. 
Let now $A$ and $B$ be linearly ordered sets, respectively by $\leq_A$ and $\leq_B$. Then, the Cartesian product $A \times B$ is linearly ordered by $\leq$ which is defined, as usual, in the following way:
$$(a,b) \leq (a',b') \quad \text{ iff } \quad ((a<_Aa') \text{ or } (a=a' \text{ and } b \leq_Bb') )$$ 
for every $(a,b),(a',b')\in A\times B$. In a similar way, the linear orders of three sets induce a linear order on their Cartesian product. 
If no confusion arises, we shall use the same symbol $\leq$ to denote every linear order considered in the sequel.  
\begin{quote}
\emph{Throughout the paper }$A$ \emph{will denote an alphabet.}
\end{quote}
A \emph{bimonoid\/ }$(K,+,\cdot,0,1)$ (cf. \cite{Dr:St}) consists of a set $K,$ two
binary operations $+$ and $\cdot$ and two constant elements $0$ and $1$ such
that $( K,+,0)$ and  $(K, \cdot, 1)$ are monoids. If the monoid $( K,+,0)$ is commutative and $0$ acts as a multiplicative zero, i.e., $k \cdot 0 = 0 \cdot k =0$ for every $k \in K$, then the bimonoid is called \emph{strong}. The bimonoid is denoted simply by $K$ if the
operations and the constant elements are understood. A semiring is a strong bimonoid where multiplication distributes over addition.
A bimonoid $K$ is called zero-sum free if $k+k'=0$ implies $k=k'=0$, and it is called zero-divisor free if $k \cdot k' =0$ implies $k=0$ or $k'=0$, for every $k,k' \in K$.

In this paper we deal with a new type of fuzzy sets with values in the Cartesian product 
$ [0,1]^4=[0,1] \times [0,1] \times [0,1] \times [0,1] $, such that their components are summing up to $1$. This type of fuzzy sets is inspired by McCarthy-Kleene logic (MK-logic for short). MK-logic which is a combination of three-valued logics of  McCarthy~\cite{Mc:Ba} and Kleene~\cite{Kl:In}, has been introduced in \cite{Ko:Fo,Av:Pr} to reason about computation errors. It is a four-valued logic, where alongside the truth values $t$ (true) and $f$ (false) there are also $u$ (undefined, which originates from Kleene's logic) and $e$ (error, which comes from McCarthy's logic). In this combination, `undefined' is intended to represent non-critical errors, while `error' is reserved for critical ones.
For the reader's convenience we recall the truth tables of MK-logic:

\medskip

$\begin{array}[t]{ll}
\begin{array}[t]{cccccccccc}
\text{or} &  \vline & t & \vline & f & \vline & u & \vline & e  \\
\hline 
	t	   &  \vline & t & \vline & t & \vline &  t & \vline & t  \\
    f	   &  \vline & t & \vline & f & \vline &  u & \vline & e  \\
	u	   &  \vline & t & \vline & u & \vline &  u & \vline & e  \\
	e	   &  \vline & e & \vline & e & \vline &  e & \vline & e  \\
\end{array}
& \qquad  \  \
\begin{array}[t]{cccccccccc}
\text{not} &  \vline & t & \vline & f & \vline & u & \vline & e  \\
\hline 
		   &  \vline & f & \vline & t & \vline &  u & \vline & e  \\
\end{array}
\end{array}$

\medskip

$
\begin{array}{ll}
\begin{array}[t]{cccccccccc}
\text{and} &  \vline & t & \vline & f & \vline & u & \vline & e  \\
\hline 
	t	   &  \vline & t & \vline & f & \vline &  u & \vline & e  \\
    f	   &  \vline & f & \vline & f & \vline &  f & \vline & f  \\
	u	   &  \vline & u & \vline & f & \vline &  u & \vline & e  \\
	e	   &  \vline & e & \vline & e & \vline &  e & \vline & e  \\
\end{array}
& \qquad  \  \
\begin{array}[t]{cccccccccc}
\text{implies} &  \vline & t & \vline & f & \vline & u & \vline & e  \\
\hline 
	t	   &  \vline & t & \vline & f & \vline &  u & \vline & e  \\
    f	   &  \vline & t & \vline & t & \vline &  t & \vline & t  \\
	u	   &  \vline & t & \vline & u & \vline &  u & \vline & e  \\
	e	   &  \vline & e & \vline & e & \vline &  e & \vline & e  \\
\end{array}
\end{array}$

\medskip

For the fuzzification of the MK-logic we assign to $t,f,u,e$ values from the interval $[0,1]$ with the restriction that they are summing up to $1$. Therefore, our fuzzy sets get their values in the subset $K$ of the Cartesian product 
$ [0,1]^4$ which is defined as follows:
$$K=\{(t,f,u,e) \in [0,1]^4 \mid t+f+u+e = 1 \}.$$
Due to practical applications, by which our theory is motivated (cf. \cite{Ku:Lo}), we refer to the four components of the elements of $K$ to as the \emph{true}, \emph{false}, \emph{unknown}, and \emph{error} value, respectively. We shall denote the elements of $K$ with bold symbols and we shall call them the \emph{truth values} of our fuzzy sets. For $\mathbf{k} =(t,f,u,e) \in K$ we shall write sometimes $x(\mathbf{k})$ for $x \in \{t,f,u,e \}$, to denote the $x$ value of $\mathbf{k}$.   
 For every $\mathbf{k_1}=(t_1, f_1, u_1, e_1),\mathbf{k_2}=(t_2, f_2, u_2, e_2) \in K$ we let $\mathbf{k_3}= \mathbf{k_1} \sqcup \mathbf{k_2}$ and $\mathbf{k_4}= \mathbf{k_1} \sqcap \mathbf{k_2}$ where $\mathbf{k_3}=(t_3, f_3, u_3, e_3)$ and $\mathbf{k_4}=(t_4, f_4, u_4, e_4)$ are defined by the relations

\smallskip

$\begin{array}[c]{ll}
t_3= t_1 + (f_1 +u_1)t_2  & \ \ \ \ \qquad t_4= t_1 t_2 \\
f_3=f_1 f_2   &  \ \ \ \ \qquad f_4=f_1 + (t_1+u_1)f_2 \\
u_3=f_1 u_2 + u_1 (f_2 + u_2) & \ \ \ \  \qquad u_4=t_1 u_2 + u_1 (t_2+u_2) \\
e_3= e_1 + (f_1 +u_1)e_2 & \ \ \ \ \qquad e_4=e_1 +(t_1 +u_1) e_2.
\end{array}$

\smallskip

\noindent We call $\sqcup$ the \emph{MK-disjunction} (\emph{disjunction} for simplicity) and $\sqcap$ the \emph{MK-conjunction} (\emph{conjunction} for simplicity). The result of the empty MK-conjunction equals $\mathbf{1}$. 
MK-disjunction and MK-conjunction correspond to the fuzzification of the connectives `or', `and' of the MK-logic, respectively. To clarify this, we preserve the above notations for $\mathbf{k_1},\mathbf{k_2},\mathbf{k_3}$, and $\mathbf{k_4}$ and construct the following multiplication table:

\begin{equation}
\begin{array}{cccccccccc}
     &  \vline & t_2 & \vline & f_2 & \vline & u_2 & \vline & e_2  \\
\hline 
	t_1	   &  \vline & t_1t_2 & \vline & t_1f_2 & \vline &  t_1u_2 & \vline & t_1e_2  \\
	f_1	   &  \vline & f_1t_2 & \vline & f_1f_2 & \vline &  f_1u_2 & \vline & f_1e_2  \\
	u_1	   &  \vline & u_1t_2 & \vline & u_1f_2 & \vline &  u_1u_2 & \vline & t_1e_2  \\
	e_1	   &  \vline & e_1t_2 & \vline & e_1f_2 & \vline &  e_1u_2 & \vline & e_1e_2  \\
\end{array}
\end{equation}

\noindent We compute every component $y_3 \in \{t_3,f_3,u_3,e_3\}$ of $\mathbf{k_3}$ by summing up the values of the cells in table (1) above, such that the corresponding cells in the truth table of `or' contain the value $y$. Similarly, for $\mathbf{k_4}$ we compute every component $y_4 \in \{t_4,f_4,u_4,e_4\}$ of $\mathbf{k_4}$ by summing up the values of the cells in table (1) above, such that the corresponding cells in the truth table of `and' contain the value $y$. For instance $t_3=t_1t_2+t_1f_2+t_1u_2+t_1e_2+f_1t_2+u_1t_2=t_1(t_2+f_2+u_2+e_2)+(f_1+u_1)t_2=t_1+(f_1+u_1)t_2$ and $t_4=t_1t_2$.

\begin{proposition}\label{oper_ass}
The disjunction and conjunction operations on $K$ are associative with unit elements $\mathbf{0}=(0,1,0,0)$ and $\mathbf{1}=(1,0,0,0)$, respectively. 
\end{proposition}

By Proposition \ref{oper_ass}, we immediately get the next corollary.

\begin{corollary}
The structure $(K, \sqcup, \sqcap, \mathbf{0}, \mathbf{1})$ is a bimonoid.  
\end{corollary} 

Nevertheless, by the following proposition we conclude that the bimonoid $(K, \sqcup, \sqcap, \mathbf{0}, \mathbf{1})$ is not strong.

\begin{proposition} \label{zero}
Both the disjunction and conjunction operations on $K$ are not commutative and idempotent. Furthermore, for every $ \mathbf{k}=(t,f,u,e) \in K$ we get $ \mathbf{0} \sqcap \mathbf{k}= \mathbf{0}$ and $\mathbf{k} \sqcap \mathbf{0}= (0, t+f+u, 0, e)$. 
\end{proposition}
\begin{proof}
Consider the elements $\mathbf{k}=(0.3, 0.2, 0.4, 0.1), \mathbf{k'}=(0.9,0.05,0.03,0.02) \in K$. Then we get $\mathbf{k} \sqcup \mathbf{k'} \neq \mathbf{k'} \sqcup \mathbf{k}$, and  $\mathbf{k} \sqcap \mathbf{k'} \neq \mathbf{k'} \sqcap \mathbf{k}$, $\mathbf{k} \sqcup \mathbf{k} \neq \mathbf{k}$ and $\mathbf{k} \sqcap \mathbf{k} \neq \mathbf{k}$. The remaining part of our proposition is proved by a standard calculation. \hfill $\square$
\end{proof} 

\begin{proposition} \label{not-distr}
Both the disjunction and conjunction on $K$ do not distribute over each other.
\end{proposition}

\begin{proposition} \label{zero-sum-div}
The bimonoid $K$ is zero-sum free and zero-divisor free.
\end{proposition}

An \emph{MK-fuzzy language over}
$A$ \emph{and} $K$ is a mapping $s:A^* \rightarrow K$. The \emph{strong support of} $s$ is the language $\mathrm{stgsupp}(s)=\{w \in A^* \mid t(s(w)) \neq 0 \}$. For every $w \in A^*$ the MK-fuzzy language $\overline{w}$ is determined by $\overline{w}(u)=\mathbf{1}$ if $u=w$, and $\overline{w}(u)=\mathbf{0}$ otherwise. The \emph{constant} MK-fuzzy language $\widetilde{\mathbf{k}}$
($\mathbf{k}\in K$) is defined, for every $w\in A^*$,$\ $by $  \widetilde
{\mathbf{k}}(w)  =\mathbf{k}$. We shall denote by $K\left\langle \left\langle A^* \right\rangle
\right\rangle $ the class of all MK-fuzzy languages over $A$ and $K$. The \emph{characteristic MK-fuzzy language}  $\mathbf{1}_L \in K\left\langle \left\langle A^* \right\rangle
\right\rangle $ \emph{of} a language $L \subseteq A^*$ is defined by $\mathbf{1}_L(w)=\mathbf{1}$ if $w \in L$ and $\mathbf{1}_L(w)=\mathbf{0}$ otherwise. 
Let $s,r\in K\left\langle \left\langle A^*\right\rangle \right\rangle $ and
$\mathbf{k}\in K$. The \emph{MK-disjunction} (or simply \emph{disjunction}) $s \sqcup r$, the \emph{MK-conjunction} (or simply \emph{conjunction}) $s \sqcap r$, and the \emph{MK-conjunctions with scalars} (simply \emph{scalar conjunctions}) $\mathbf{k} \sqcap s$ and $s \sqcap \mathbf{k}$ are defined as follows:
$s \sqcup r(w)=s(w) \sqcup r(w)$, $s \sqcap r(w)=s(w) \sqcap r(w)$, and $(\mathbf{k} \sqcap s)(w)=\mathbf{k} \sqcap s(w)$, $(s \sqcap \mathbf{k})(w)=s(w)\sqcap \mathbf{k}$
for every $w \in A^*$. Since the disjunction and conjunction operations among MK-fuzzy languages are defined elementwise, we can easily show that properties of the structure $\left(
K\left\langle \left\langle A^*\right\rangle \right\rangle ,\sqcup,\sqcap, \widetilde{\mathbf{0}}, \widetilde{\mathbf{1}} \right)$ are inherited by the properties of the structure  $(K, \sqcup, \sqcap, \mathbf{0}, \mathbf{1})$,  hence $\left(
K\left\langle \left\langle A^*\right\rangle \right\rangle ,\sqcup,\sqcap, \widetilde{\mathbf{0}}, \widetilde{\mathbf{1}} \right)  $ is a bimonoid. 
The \emph{Cauchy product} $rs$ \emph{of} $r,s \in K\left\langle\left\langle A^* \right\rangle\right\rangle$ is defined as follows. For every $w=a_0 \ldots a_{n-1} \in A^*$ with $a_0, \ldots, a_{n-1} \in A$ we let

$rs(w)=\left( r(\varepsilon ) \sqcap s(a_0 \ldots a_{n-1}) \right) \sqcup \left (r(a_0) \sqcap s(a_1 \ldots a_{n-1}) \right) \sqcup \ldots \sqcup \left( r(a_0 \ldots a_{n-1}) \sqcap s(\varepsilon) \right).$  

\noindent Since disjunction and conjunction are not commutative, and they do not distribute over each other, the Cauchy product is not associative as we state in the next proposition.

\begin{proposition}\label{Cauchy}
The Cauchy product operation is not associative.
\end{proposition}

We assume now that the alphabet $A$ is linearly ordered and let $B$ be another alphabet. Then a homomorphism $h:A^{\ast}\rightarrow B^{\ast}$ is extended to a mapping $h:K\left\langle \left\langle A^*\right\rangle \right\rangle \rightarrow K\left\langle \left\langle B^*\right\rangle \right\rangle$ in the following way. For every $s \in K\left\langle \left\langle A^*\right\rangle \right\rangle$ and $u \in B^*$ we let $h(s)(u)= \bigsqcup_{w \in h^{-1}(u)}s(w)$ where in the definition of the disjunction we take into account the lexicographic order of the words $w \in h^{-1}(u)$. 
Finally, we assume that $h:A^{\ast}\rightarrow B^{\ast}$ is a
strict alphabetic homomorphism, i.e., $h(a)\in B$ for every $a\in A$. Then, for every $r\in
K\left\langle \left\langle B^*\right\rangle \right\rangle $ the
MK-fuzzy language $h^{-1}(r)\in K\left\langle \left\langle A^*\right\rangle
\right\rangle $ is determined by $h^{-1}(r)(w)=r(h(w))$ for
every $w\in A^*$. We should note that for $h^{-1}$ we do not require any order on the alphabet $A$.

\section{MK-fuzzy automata}

In this section we introduce the model of MK-fuzzy automata over $A$ and $K$ and investigate closure properties of the class of their behaviors. Moreover, we prove a Nivat-like theorem for recognizable MK-fuzzy languages. 

\begin{definition}
An \emph{MK-fuzzy automaton over} $A$ \emph{and} $K$ is a seven-tuple $\mathcal{A}=(Q,I,T,F,in, wt, \ter)$ where $Q$ is the \emph{finite state set} which is assumed to be linearly ordered, $I$ is the \emph{set of initial states}, $T \subseteq Q \times A \times Q$ is the \emph{set of transitions}, $F$ is the \emph{set of  final states}, $in: I \rightarrow K $ is the \emph{initial distribution}, $wt: T \rightarrow K$ is a mapping assigning \emph{truth values} to the transitions of the automaton, and $\ter : F \rightarrow K$ is the \emph{final distribution}. 
\end{definition}

Let $w=a_{0} \ldots a_{n-1}$ be a word over $A$ with $a_0, \ldots, a_{n-1} \in A$. A \emph{path} $P^{\left(\mathcal{A}\right )}_w$ (or simply $P_w$ if the automaton is understood) \emph{of} $\mathcal{A}$ \emph{over} $w$ is a sequence of transitions  $P^{(\mathcal{A})}_w : =\left (\left ( q_i, a_i,q_{i+1} \right ) \right )_{0 \leq i \leq n-1}$, $(q_i,a_i,q_{i+1}) \in T$ for every $0 \leq i \leq n-1$, with $q_0 \in I$ and $q_n \in F$. The \emph{weight of} $P^{(\mathcal{A})}_w$ is the truth value 
\[ 
weight\left(P^{(\mathcal{A})}_w\right)= in(q_0) \sqcap  \bigsqcap_{0 \leq i \leq n-1} wt \left ( q_i, a_i, q_{i+1} \right )  \sqcap \ter(q_n).
\] 
The set of paths of $\mathcal{A}$ over $w$ can be linearly ordered as follows. For two paths $P_w  =\left (\left ( q_i, a_i,q_{i+1} \right ) \right )_{0 \leq i \leq n-1}$ and $P'_w  =\left (\left ( q'_i, a_i,q'_{i+1} \right ) \right )_{0 \leq i \leq n-1}$ we let 
$$P_w \leq P'_w \quad \text{   iff   }\quad q_0 \ldots q_{n-1} \leq_{lex} q'_0 \ldots q'_{n-1}.$$
The \emph{behavior of} $\mathcal{A}$ is the MK-fuzzy language $\| \mathcal{A} \| : A^* \rightarrow K$ and it is defined in the following way. Let $w \in A^+$ and $\{P_{w,1}, \ldots , P_{w,m} \}$ be the set of all paths of $\mathcal{A}$ over $w$. Furthermore, assume that $P_{w,1} \leq \ldots \leq P_{w,m}$. Then, we set
$$\| \mathcal{A} \|(w)= weight(P_{w,1}) \sqcup \ldots \sqcup weight(P_{w,m}).$$
If there are no paths of $\mathcal{A}$ over $w$, then we let $\| \mathcal{A} \|(w)= \mathbf{0}$. If $w=\varepsilon$, then 
$$\| \mathcal{A} \|(\varepsilon)=(in(q_{i_1}) \sqcap \ter(q_{i_1})) \sqcup \ldots \sqcup (in(q_{i_m}) \sqcap \ter(q_{i_m})) $$
where $I \cap F=\{q_{i_1}, \ldots, q_{i_m} \}$ and $q_{i_1} \leq \ldots \leq q_{i_m}$. If $I \cap F =\emptyset$, then we set $\| \mathcal{A} \|(\varepsilon)=\mathbf{0}$.
An MK-fuzzy language $s :A^* \rightarrow K$ is called \emph{recognizable} if there is an MK-fuzzy automaton $ \mathcal{A}$ over $A$ and $K$ such that $s= \| \mathcal{A} \| $. We denote by $Rec(K,A)$ the class of all recognizable MK-fuzzy languages over $A$ and $K$.

\begin{remark}
By our definition above, we get that $weight\left(P^{(\mathcal{A})}_w \right)=\mathbf{0}$ whenever $in(q_0)=\mathbf{0}$ for every path $P^{(\mathcal{A})}_w  =\left (\left ( q_i, a_i,q_{i+1} \right ) \right )_{0 \leq i \leq n-1}$ of $\mathcal{A}$ over $w=a_0 \ldots a_{n-1}$. Hence, in the sequel, we assume that $in:I\rightarrow K \setminus\{\mathbf{0} \}$ for every MK-fuzzy automaton $\mathcal{A}=\{Q,I,T,F,in,wt,\ter \}$ over $A$ and $K$.
\end{remark}

\begin{example}\label{constant}
Let $\mathbf{k} \in K$. Then the constant MK-fuzzy language $\widetilde{\mathbf{k}}$ is recognizable. Indeed, we consider the MK-fuzzy automaton $\mathcal{A}_{\mathbf{k}}=(\{q\},\{q\},T,\{q\},in,wt,\ter)$ with $T=\{(q,a,q) \mid a \in A \}$ and $in(q)= \mathbf{k}$, $\ter(q)=\mathbf{1}$, and $wt(q,a,q)=\mathbf{1}$ for every $a \in A$. We trivially get $\|\mathcal{A} \|= \widetilde{\mathbf{k}}$.

\end{example}

\begin{proposition}\label{rec-char_rec}
Let $L \subseteq A^*$ be a recognizable language. Then $\mathbf{1}_L \in Rec(K,A)$. 
\end{proposition}

\begin{theorem}\label{Rec-disjunction}
The class $Rec(K,A)$ is closed under disjunction.
\end{theorem}

\begin{theorem}\label{conj1}
Let $s \in Rec(K,A)$ and $L \subseteq A^*$ be a recognizable language. Then $\mathbf{1}_L \sqcap s \in Rec(K,A)$.
\end{theorem}
\begin{proof} [Sketch] 
Let $\mathcal{A}_1=\left(Q_1,A,q^{(0)}_1,T_1,F_1\right)$ be a deterministic finite automaton accepting $L$ and $\mathcal{A}_2=(Q_2,I_2,  T_2,F_2,in_2, wt_2,\ter_2)$ an MK-fuzzy automaton over $A$ and $K$ accepting $s$. We define an arbitrary linear order $\leq$ on $Q_1$ and consider the MK-fuzzy automaton $\mathcal{A}=(Q_1 \times Q_2,\{q^{(0)}_1 \} \times I_2, T, F_1 \times F_2, in,wt,  \ter  )$ with $T= \{ ((q_1,q_2),a,(q'_1,q'_2)) \mid (q_1,a,q'_1) \in T_1 \text{ and } (q_2,a,q'_2) \in T_2 \}$ and 
\begin{itemize}
\item[-] $in\left(q^{(0)}_1,q_2\right)=in_2(q_2)$ for every $q_2 \in I_2$,
\item[-] $wt((q_1,q_2),a,(q'_1,q'_2))=wt_2(q_2,a,q'_2)$ for every $((q_1,q_2),a,(q'_1,q'_2)) \in T$,
\item[-] $\ter(q_1,q_2)=\ter_2(q_2)$ for every $(q_1,q_2) \in F_1 \times F_2$.
\end{itemize}

\noindent The state set $Q_1 \times Q_2$ is linearly ordered by 
$$(q_1,q_2) \leq (q'_1,q'_2) \qquad  \text{ iff } \qquad  ((q_2 < q'_2) \text{ or }(q_2 = q'_2 \text{ \ and \ } q_1 \leq q'_1)) $$
for every $(q_1,q_2),(q'_1,q'_2) \in Q_1 \times Q_2$.  Then we  show that $\|\mathcal{A}\|=\mathbf{1}_L \sqcap s$. \hfill $\square$
\end{proof}

\begin{theorem}\label{hom}
Let $A$ be a linearly ordered alphabet and $h: A^* \rightarrow B^*$ a strict alphabetic homomorphism. Then $s \in Rec(K,A)$ implies $h(s) \in Rec(K,B)$. 
\end{theorem}
\begin{proof}
Let $\mathcal{A}=(Q,I,T,F,in,wt,\ter)$ be an MK-fuzzy automaton over $A$ and $K$ accepting $s$. We consider the MK-fuzzy automaton $\mathcal{B}=(A\times Q, \{\min A\}\times I,T',A\times F,in',wt',\ter')$ with $T'=\{((a,q),b,(a'q')) \mid (q,a',q') \in T \text{ and } h(a')=b  \}$. The weight mappings $in', wt', \ter'$ are defined respectively, by 
\begin{itemize}
\item[-] $in'(\overline{a},q)=in(q)$, with $\overline{a}=\min A$ and every $q \in I$,
\item[-] $wt'((a,q),b,(a',q'))=wt(q,a',q')$, for every $((a,q),b,(a',q')) \in T'$, and
\item[-] $\ter'(a,q)=\ter(q)$, for every $(a,q) \in A \times F$.
\end{itemize}

\noindent Let $w =a_0 \ldots a_{n-1}  \in A^+$ and $P_w^{(\mathcal{A})}=\left ( (q_i,a_i,q_{i+1})\right)_{0 \leq i \leq n-1}$ be a path of $\mathcal{A}$ over $w$. By definition of the MK-fuzzy automaton $\mathcal{B}$ there is a unique path 
$$P_{h(w)}^{(\mathcal{B})}= ((\overline{a},q_0),h(a_0),(a_0,q_1)) ((a_0,q_1),h(a_1),(a_1,q_2)) \ldots ((a_{n-2},q_{n-1}),h(a_{n-1}),(a_{n-1},q_n))$$
of $\mathcal{B}$ over $h(w)$, and by a straightforward calculation we get $weight \left(P_{h(w)}^{(\mathcal{B})} \right)=weight \left(P_w^{(\mathcal{A})} \right)$. Conversely, let $u=b_0 \ldots  b_{n-1} \in B^+$ and 
$$P_u^{(\mathcal{B})}= ((\overline{a},q_0),b_0,(a_0,q_1)) ((a_0,q_1),b_1,(a_1,q_2)) \ldots ((a_{n-2},q_{n-1}),b_{n-1},(a_{n-1},q_n))$$
be a path of $\mathcal{B}$ over $u$. Then, $u=h(w)$ where $w=a_0 \ldots a_{n-1} \in A^+$. Moreover, $P_w^{(\mathcal{A})}=\left ( (q_i,a_i,q_{i+1})\right)_{0 \leq i \leq n-1}$ is a path of $\mathcal{A}$ over $w$ and $weight \left(P_u^{(\mathcal{B})} \right)=weight \left(P_w^{(\mathcal{A})} \right)$. Hence, for every $u\in B^+$, if $w_1, \ldots, w_m $ are all the words in $A^+$ such that $h(w_i)=u$ ($1 \leq i \leq m$), then there is a one-to-one correspondence between the paths 
$$P_{w_1,1}^{(\mathcal{A})}, \ldots, P_{w_1,j_1}^{(\mathcal{A})},  \ldots, P_{w_m,1}^{(\mathcal{A})}, \ldots, P_{w_m,j_m}^{(\mathcal{A})}  $$
of $\mathcal{A}$, respectively over $w_1, \ldots, w_m$, and the paths 
$$P_{u,1}^{(\mathcal{B})}, \ldots, P_{u,j_1}^{(\mathcal{B})}, P_{u,j_1+1}^{(\mathcal{B})}, \ldots, P_{u,j_1+j_2}^{(\mathcal{B})}, \ldots, P_{u,k}^{(\mathcal{B})} $$
of $\mathcal{B}$ over $u$, where $P_{w_l,r_l}^{(\mathcal{A})}$ corresponds to $P_{u,j_1+ \ldots +j_{l-1}+r_l}^{(\mathcal{B})}$ for every $1 \leq l \leq m$ and $1 \leq r_l \leq j_l$. Then we get $weight\left(P_{w_l,r_l}^{(\mathcal{A})} \right) = weight\left( P_{u,j_1+ \ldots +j_{l-1}+r_l}^{(\mathcal{B})} \right ) $. Moreover, if $w_1 \leq \ldots \leq w_m$ and
$$P_{w_1,1}^{(\mathcal{A})} \leq  \ldots \leq P_{w_1,j_1}^{(\mathcal{A})},  \ldots,   P_{w_m,1}^{(\mathcal{A})} \leq  \ldots \leq P_{w_m,j_m}^{(\mathcal{A})},  $$
then 
$$P_{u,1}^{(\mathcal{B})} \leq \ldots \leq P_{u,j_1}^{(\mathcal{B})} \leq P_{u,j_1+1}^{(\mathcal{B})} \leq \ldots \leq P_{u,j_1+j_2}^{(\mathcal{B})} \leq \ldots \leq P_{u,k}^{(\mathcal{B})}.$$ 
Hence we have
\begin{align*}
h(s)(u) & = \bigsqcup_{w \in h^{-1}(u)}s(w)  = s(w_1) \sqcup \ldots \sqcup s(w_m) = \bigsqcup_{1 \leq r_1 \leq j_1}weight \left ( P_{w_1,r_1}^{(\mathcal{A})} \right ) \sqcup \ldots \sqcup \bigsqcup_{1 \leq r_m \leq j_m}weight \left ( P_{w_m,r_m}^{(\mathcal{A})} \right ) \\
&= \bigsqcup_{1 \leq i \leq k} weight \left ( P_{u,i}^{(\mathcal{B})} \right) = \| \mathcal{B} \|(u).
\end{align*}
If $s(\varepsilon) \neq \mathbf{0}$, then let $I\cap F =\{q_{i_1}, \ldots, q_{i_m}  \}$. Then $(\{\min A\} \times I) \cap (A \times F) = \{(\min A,q_{i_1}), \ldots,   (\min A,q_{i_m})  \}$ and by definition of $in'$ and $\ter'$ we get $\|\mathcal{A}\|(\varepsilon)=\|\mathcal{B}\|(\varepsilon)$. Since $h(s)(\varepsilon)=s(\varepsilon)$, we finally conclude that  $h(s)=\|\mathcal{B}\|$, i.e, $h(s) \in Rec(K,B)$, and we are done. \hfill $\square$
\end{proof} 

\begin{theorem}\label{inv-hom}
Let $h: A^* \rightarrow B^*$ be a strict alphabetic homomorphism. Then $s \in Rec(K,B)$ implies $h^{-1}(s) \in Rec(K,A)$. 
\end{theorem}

Next, we show a Nivat-like decomposition theorem for recognizable MK-fuzzy languages. The fundamental Nivat's theorem \cite{Ni:Tr} states a relation among rational transductions and rational languages. A Nivat-like result was proved for weighted automata over semirings in \cite{Dr:Au}. We need some preliminary matter. 
Let $B$ be an alphabet and  $g:B \rightarrow K$ a mapping. Then $g$ can be extended to an MK-fuzzy language $g:B^* \rightarrow K$ by $g(b_0 \ldots b_{n-1})=\bigsqcap_{0 \leq i \leq n-1}g(b_i)$ for every $b_0 \ldots b_{n-1} \in B^+$, $b_0, \ldots, b_{n-1}\in B$, and $g(\varepsilon)= \mathbf{1}$. Then, for a language $L \subseteq B^+$ we define the MK-fuzzy language $L \cap g$ by $L \cap g(w)=g(w)$ if $w\in L$ and $L \cap g(w)= \mathbf{0}$ otherwise, for every $w \in B^*$. It should be clear that $L \cap g= \mathbf{1}_L \sqcap g$. Now we are ready to state our Nivat-like theorem.
\begin{theorem}
Let $A$ be a linearly ordered alphabet and $s$ an MK-fuzzy language over $A$ and $K$ with $s(\varepsilon)=\mathbf{0}$. Then $s$ is recognizable iff there is a linearly ordered alphabet $B$, a recognizable language $L \subseteq B^+$, a mapping $g:B \rightarrow K$, and a strict alphabetic homomorphism $h:B^* \rightarrow A^*$ such that $s=h(L\cap g)$. 
\end{theorem}    
\begin{proof}
We prove firstly the implication ``$\impliedby$''. The  MK-fuzzy language $g$ is recognizable. Indeed, consider the MK-fuzzy automaton $\mathcal{G}=(\{q\},\{q\},T,\{q\},in,wt,\ter)$ over $B$ and $K$, with $in(q)=\ter(q)= \mathbf{1}$ and $wt(q,b,q)=g(b)$ for every $b\in B$. Trivially $\|\mathcal{G}\|=g$. Then, by Proposition \ref{rec-char_rec} and Theorem \ref{conj1} the MK-fuzzy language $\mathbf{1}_L \sqcap g$ is recognizable and hence, $h(L \cap g)$ is recognizable by Theorem \ref{hom}.

Conversely, let $s\in Rec(K,A)$ with $s(\varepsilon)=\mathbf{0}$ and $\mathcal{A}=(Q,I,T,in, wt, \ter)$ be an MK-fuzzy automaton accepting $s$. 
We set $B=T$ and consider the finite automaton $\mathcal{B}=(Q,B,I,T',F)$ with $T'=\{(q,(q,a,q'),q') \\  \mid (q,a,q')\in T \}$. It can be easily seen that $L(\mathcal{B})=\{P_w \mid w \in A^+ \text{ and } P_w \text{ path of } \mathcal{A} \text{ over } w \} \cup C$, where $C=\{\varepsilon \}$ if $I \cap F \neq \emptyset$ and $C=\emptyset$ otherwise. We let $L=L(\mathcal{B})\setminus \{\varepsilon \}$, define the mapping $g:B \rightarrow K$ by $g(q,a,q')=wt(q,a,q')$ for every $(q,a,q') \in B$, and the strict alphabetic homomorphism $h:B^* \rightarrow A^*$ by $h(q,a,q')=a$ for every $(q,a,q') \in B$. Then, for every $w\in A^+$ we get
\begin{align*}
h(L \cap g)(w) & = \bigsqcup_{u \in h^{-1}(w)}L\cap g(u)  = \underset{u \in L}{\underset{u \in h^{-1}(w)}{\bigsqcup}} g(u) = \bigsqcup_{P_w}weight(P_w)  = \|\mathcal{A}\|(w), 
\end{align*} 
i.e., $h(L \cap g)(w)=\|\mathcal{A}\|$ as required, and our proof is completed. \hfill $\square$
\end{proof}

In the sequel, we deal with the deterministic counterpart of our model. An MK-fuzzy automaton $\mathcal{A}=(Q,I,T,F, in,wt, \ter)$ over $A$ and $K$ is called \emph{deterministic} if $I=\{q_0 \} $ and for every $q \in Q, a\in A$ there is at most one $q' \in Q$ such that $(q,a,q') \in T$. Then for every word $w \in A^*$ there is at most one path $P_w$ of $\mathcal{A}$ over $w$, which in turn implies that we can relax the order relation of $Q$. Nevertheless, in the sequel, sometimes we will need the state set of a deterministic MK-fuzzy automaton to be ordered. A deterministic MK-fuzzy automaton $\mathcal{A}$ is simply written as $\mathcal{A}=(Q, q_0,T, F, in, wt,\ter)$. 
An MK-fuzzy language $s \in K\left\langle \left\langle A^*\right\rangle \right\rangle$ is called \emph{deterministically recognizable} if there is a deterministic MK-fuzzy automaton $\mathcal{A}$ over $A$ and $K$ such that $s= \| \mathcal{A} \| $. We denote by $DRec(K,A)$ the class of all deterministically recognizable MK-fuzzy languages over $A$ and $K$.  
An MK-fuzzy automaton  $\mathcal{A}=(Q,I,T,F, in,wt,\ter)$ is called \emph{unambiguous} if for every word $ w \in A^*$ there is at most one path $P_w$ of $\mathcal{A}$ over $A$. Clearly, every deterministic MK-fuzzy automaton is unambiguous as well, but the converse is not always true.

\begin{theorem}\label{DRec-scalars}
Let $s \in DRec(K,A)$ and $\mathbf{k} \in K$. Then $\mathbf{k}  \sqcap s, s \sqcap \mathbf{k} \in DRec(K,A).$ 
\end{theorem}

Next, we investigate the closure of the class of deterministically recognizable MK-fuzzy languages under Cauchy product. More precisely, we show that the Cauchy product of two deterministically recognizable MK-fuzzy languages is a  recognizable MK-fuzzy language. For this, we will need the notion of a normalized MK-fuzzy automaton and some preliminary results which present their own interest. 

\begin{definition}
An MK-fuzzy automaton $\mathcal{A}=(Q, I,T, F, in, wt,\ter)$ is called \emph{normalized} if $I=\{ q_{in} \}$, $q_{in} \notin F$, $in(q_{in})=\mathbf{1}$, $\ter(q)=\mathbf{1}$ for every $q \in F$, $(q,a,q_{in}) \notin T$ for every $q \in Q, a \in A$, and $(q,a,q') \notin T$ for every $q \in F, a\in A$, and $q' \in Q$. 
\end{definition}

By the above definition, if $\mathcal{A}$ is a normalized MK-fuzzy automaton, then $\| \mathcal{A} \| (\varepsilon) = \mathbf{0}$. A normalized MK-fuzzy automaton $\mathcal{A}=(Q, I,T, F, in, wt,\ter)$ will be simply denoted by $\mathcal{A}=(Q, q_{in},T, F, wt)$. 

\begin{proposition} \label{normalization}
For every deterministic MK-fuzzy automaton $\mathcal{A}=(Q, q_0,T, F, in, wt, \ter)$ we can effectively construct a normalized unambiguous MK-fuzzy automaton $\mathcal{A}'$ such that $\|\mathcal{A}'\|(w)=\|\mathcal{A}\|(w)$ for every $w \in A^+$, and $\|\mathcal{A}'\|(\varepsilon)=\mathbf{0}$.
\end{proposition}   

\begin{lemma}\label{unambiguous-scalars}
Let $s \in K\left\langle \left\langle A^* \right\rangle
\right\rangle $ and $\mathbf{k} \in K$. If $s$ is accepted by a normalized unambiguous MK-fuzzy automaton, then $s \sqcap \mathbf{k}$ is accepted also by a normalized unambiguous MK-fuzzy automaton. 
\end{lemma}

\begin{theorem}\label{det-Cauchy}
Let $r,s \in DRec(K,A)$. Then $rs \in Rec(K,A)$.
\end{theorem}
\begin{proof} [Sketch]
Since $r,s \in DRec(K,A)$, there are deterministic MK-fuzzy automata accepting them. Then, by Proposition \ref{normalization}, we can effectively construct normalized unambiguous MK-fuzzy automata $\mathcal{A}_1=\left(Q_1, q^{(1)}_{in},T_1, F_1, wt_1 \right)$ and $\mathcal{A}_2=\left(Q_2, q^{(2)}_{in},T_2, F_2, wt_2 \right)$ such that $\|\mathcal{A}_1\|(w)=r(w)$ and $\|\mathcal{A}_2\|(w)=s(w)$ for every $w \in A^+$. Without any loss we assume that $Q_1 \cap Q_2 = \emptyset$, otherwise we apply a renaming. We consider the MK-fuzzy automaton $\mathcal{A}=\left (Q, \{q^{(1)}_{in}\}, T, F_2, in, wt, \ter  \right)$ with
\begin{itemize}
\item[-] $Q= (Q_1 \setminus F_1) \cup Q_2$,
\item[-] $T=\left \{\left(q^{(1)},a,p^{(1)}\right) \in T_1 \mid p^{(1)} \notin F_1 \right \} \cup T_2 \cup \\  \left \{\left(q^{(1)},a,q^{(2)}_{in}\right) \mid \text{ there exists } p^{(1)} \in F_1  \text{ such that } \left(q^{(1)},a,p^{(1)} \right) \in T_1 \right \}$,
\item[-] $in\left(q^{(1)}_{in} \right)=\mathbf{1}$, 
\item[-] $wt(q,a,p)=\left\{
\begin{array}
[c]{ll}%
wt_1(q,a,p) & \text{if } (q,a,p) \in T_1\\
wt_2(q,a,p)  & \text{if } (q,a,p) \in T_2 \\
wt_1(q,a,p^{(1)})  & \text{if }  q \in Q_1 \setminus F_1, \  p=q^{(2)}_{in}, \ p^{(1)} \in F_1, \text{ and } \left(q,a,p^{(1)}\right) \in T_1 
\end{array}
\text{ }\right.  $ \smallskip \\
for every $(q,q,p) \in T$, and
\item[-] $\ter(q)=\mathbf{1}$ for every $q \in F_2$.
\end{itemize}
We should note that in case $p=q^{(2)}_{in}$ the value $wt(q,a,p)$ is well-defined. Indeed, since the original MK-fuzzy automaton accepting $r$ is deterministic, by construction of $\mathcal{A}_1$, we get that there is at most one $p^{(1)} \in F_1$ such that $(q,a,p^{(1)}) \in T_1$. 
We define a linear order on $Q$ by preserving the orders of $Q_1$ and $Q_2$ and letting $\max Q_2 \leq \min Q_1$. Then we can show that $\|\mathcal{A}\|(w)=rs(w)$ for every $w \in A^+$.

Next, by Theorem \ref{DRec-scalars}, the series $r(\varepsilon) \sqcap s$ is deterministically recognizable, hence by Proposition \ref{normalization} there is a normalized unambiguous MK-fuzzy automaton $\mathcal{A}_3$ such that $\|\mathcal{A}_3\|(w)=( r(\varepsilon) \sqcap s)(w)$ for every $w \in A^+$, and $\|\mathcal{A}_3\|(\varepsilon)= \mathbf{0}$. Furthermore, by Proposition \ref{rec-char_rec} and Lemma \ref{unambiguous-scalars} respectively, the MK-fuzzy languages $\overline{\varepsilon} \sqcap r(\varepsilon) \sqcap s(\varepsilon) $ and  $\| \mathcal{A}_1  \| \sqcap s(\varepsilon)$ are recognizable. Since
$$rs= (\overline{\varepsilon} \sqcap r(\varepsilon) \sqcap s(\varepsilon) ) \sqcup \| \mathcal{A}_3 \|  \sqcup \|\mathcal{A}\| \sqcup (\| \mathcal{A}_1  \| \sqcap s(\varepsilon)),$$ 
we conclude our proof by Theorem \ref{Rec-disjunction}. \hfill $\square$
\end{proof} 

\begin{proposition}\label{strong-supp}
Let $s \in DRec(K,A)$. Then the strong support of $s$ is a recognizable language.
\end{proposition}

\section{MK-fuzzy monadic second order logic}
In this section we introduce our MK-fuzzy monadic second order (MSO for short) logic and we prove the fundamental theorem of B\"{u}chi \cite{Bu:We}, Elgot \cite{El:De}, and Trakhtenbrot \cite{Tr:Fi} in the setup of MK-fuzzy languages. We need to recall the definition of syntax and semantics of MSO logic (cf. for instance \cite{Th:La}).

The syntax of MSO logic formulas over $A$ is given by the grammar
\[
\phi::=\mathit{true}\mid P_{a}(x)\mid x\leq x'\mid  x\in X \mid \lnot\phi\mid\phi\vee\phi\mid\exists x\centerdot\phi \mid \exists X\centerdot\phi
\]
where $a\in A$ and we let $\mathit{false}=\lnot \mathit{true}$. The set $\mathit{free}(\phi)$ of free
variables of an MSO logic formula $\phi$ is defined as usual. In order
to define the semantics of MSO logic formulas we need the notions of the extended alphabet  and valid assignment. Let $\mathcal{V}$ be a
finite set of first and second order variables. For every word $w=w(0)\ldots w(n-1) \in A^*$ we let $dom(w)=\left\{  0,\ldots,n-1\right\}
$. A $(\mathcal{V},w)$-\emph{assignment} $\sigma$ is a mapping associating first order variables
from $\mathcal{V}$ to elements of $dom(w)$, and second order variables from
$\mathcal{V}$ to subsets of $dom(w)$. If $x$ is a first order variable and
$i\in dom(w),$ then $\sigma\lbrack x\rightarrow i]$ denotes the $(\mathcal{V}\cup\{x\},w)$-assignment which associates $i$ to $x$ and coincides with
$\sigma$ on $\mathcal{V}\setminus\{x\}$. For a second order variable $X$ and
$I\subseteq dom(w),$ the notation $\sigma\lbrack X\rightarrow I]$ has a
similar meaning.
We shall encode pairs of the form $(w,\sigma)$, where $w\in A^{\ast}$ and
$\sigma$ is a $(\mathcal{V},w)$-assignment, using the extended alphabet
$A_{\mathcal{V}}=A \times\{0,1\}^{\mathcal{V}}$. Indeed, every word
in $A_{\mathcal{V}}^{\ast}$ can be considered as a pair $(w,\sigma)$
where $w$ is the projection over $A$ and $\sigma$ is the projection over
$\{0,1\}^{\mathcal{V}}.$ Then $\sigma$ is a valid assignment if for every first
order variable $x\in\mathcal{V}$ the $x$-row contains exactly one $1$. In this
case, $\sigma$ is the $(\mathcal{V},w)$-assignment such that for every first
order variable $x\in\mathcal{V}$, $\sigma(x)$ is the position of the $1$ on
the $x$-row, and for every second order variable $X\in\mathcal{V},$ $\sigma(X)$
is the set of positions labelled with $1$ along the $X$-row.
It is well-known that 
\[
N_{\mathcal{V}}=\{(w,\sigma)\in A_{\mathcal{V}}^{\ast}\mid\sigma\text{ is
a valid }(\mathcal{V},w)\text{-assignment}\}
\]
is a recognizable language. 
For every $(w,\sigma
)\in\mathcal{N}_{\mathcal{V}}$ we define the satisfaction relation $(w,\sigma)\models\phi$ by
induction on the structure of $\phi$, as follows:

\medskip

$\begin{array}[c]{ll}
(w,\sigma)\models true,  & \ \ \ \qquad (w,\sigma)\models x\in X  \ \text{ iff } \ \sigma(x)\in \sigma(X), \\
(w,\sigma)\models P_{a}(x)  \ \text{ iff } \ w(\sigma(x))=a, & \ \ \ \qquad (w,\sigma)\models\lnot\phi \ \text{ iff } \ (w,\sigma)\not \models \phi, \\
(w,\sigma)\models x\leq x' \ \text{ iff } \ \sigma(x)\leq\sigma(x'), & \ \ \ \qquad (w,\sigma)\models\phi\vee\phi^{\prime} \ \text{ iff } \ (w,\sigma
)\models\phi \text{ or } (w,\sigma)\models\phi^{\prime}, 
\end{array}$

$\begin{array}[c]{ll}
(w,\sigma)\models\exists x\centerdot\phi \ \text{ iff  there exists an } i\in dom(w)\ \text{ such that } (w,\sigma\lbrack x\rightarrow i\rbrack )\models\phi, \\
(w,\sigma)\models\exists X\centerdot\phi \ \text{ iff there exists an }
I\subseteq dom(w)\ \text{ such that } (w,\sigma\lbrack X\rightarrow I \rbrack)\models\phi.
\end{array}$

\medskip

\noindent If $(w,\sigma)\in A_{\mathcal{V}}^{\ast}\setminus\mathcal{N}_{\mathcal{V}}$, then we let $(w,\sigma)\not \models \phi$. 

We denote by $L(\phi)$ the language of an MSO logic sentence $\phi$, i.e., $L(\phi)=\{w \in A^* \mid w \models \phi \}$. 
\

\begin{remark}\label{order}
For the definition of the semantics of our MK-fuzzy \emph{MSO} logic, we shall need the power set $\mathcal{P}(dom(w))$ to be linearly ordered for every word $w \in A^*$.  Let $w=a_0 \ldots a_{n-1} \in A^*$, hence $dom(w)=\{0, \ldots, n-1 \}$. We define the linear order $\leq $ on $\mathcal{P}(dom(w))$ in the following way. Let $I=\{i_1,\ldots,i_m \}, J=\{j_1, \ldots, j_k \} \in \mathcal{P}(dom(w))$ and assume that $0 \leq i_1 < \ldots < i_m \leq n-1$ and $0 \leq j_1 < \ldots < j_k \leq n-1$. Then we consider the words $u_I=i_1 \ldots i_m, u_J=j_1 \ldots j_k \in dom(w)^*$. Clearly, there is a one-to-one correspondence among the subsets of $dom(w)$, and the words of $dom(w)^*$ with length at most $n$ and their letters being pairwise disjoint. The empty set corresponds to the empty word. Now, for every $I, J \in \mathcal{P}(dom(w))$ we set $I \leq J$ iff $u_I \leq_{lex} u_J$.   
\end{remark}

\begin{definition}
The syntax of formulas of the \emph{MK-fuzzy MSO logic over } $A$
\emph{and} $K$ is given by the grammar
\begin{align*}
\varphi &  ::=\mathbf{k} \mid\phi\mid\varphi \oplus \varphi\mid\varphi\otimes\varphi\mid{\textstyle\bigoplus\nolimits_{x}}
\centerdot\varphi\mid{\textstyle\bigoplus\nolimits_{X}}
\centerdot\varphi\mid{\textstyle\bigotimes\nolimits_{x}}\centerdot\varphi
\end{align*}
where $\mathbf{k} \in K$, $a\in A$, and $\phi$ denotes an \emph{MSO} logic formula.
\end{definition}

We denote by $MSO(K,A)$ the set of all MK-fuzzy MSO logic formulas
$\varphi$ over $A$ and $K$. We represent the semantics of formulas
$\varphi\in MSO(K,A)$ as MK-fuzzy languages $\left\Vert \varphi\right\Vert \in
K\left\langle \left\langle A^{\ast}\right\rangle \right\rangle $. For the
semantics of MSO logic formulas $\phi$ we use the satisfaction relation
as defined above. Therefore, the semantics of MSO logic formulas $\phi$
gets only the values $\mathbf{0}$ and $\mathbf{1}$.

\begin{definition}
\label{semantics}Let $\varphi\in MSO(K,A)$ and $\mathcal{V}$ be a finite
set of variables with $\mathit{free}(\varphi)\subseteq\mathcal{V}$. The semantics of
$\varphi$ is an MK-fuzzy language $\left\Vert \varphi\right\Vert _{\mathcal{V}}\in
K\left\langle \left\langle A_{\mathcal{V}}^{\ast}\right\rangle
\right\rangle$. Consider an element $(w,\sigma)\in A_{\mathcal{V}}^{\ast
}$. If $(w,\sigma)\notin N_{\mathcal{V}}$, then we let $  \left\Vert
\varphi\right\Vert _{\mathcal{V}}(w,\sigma)  =\mathbf{0}.$ Otherwise, we define
$  \left\Vert \varphi\right\Vert _{\mathcal{V}}(w,\sigma)  \in
K$, inductively on the structure of $\varphi$, as follows:
\begin{itemize}
\item[-] $\left\Vert \mathbf{k} \right\Vert _{\mathcal{V}}(w,\sigma)=\mathbf{k},$ 
\item[-] $\left\Vert \phi\right\Vert _{\mathcal{V}}(w,\sigma)=\left\{
\begin{array}
[c]{rl}\mathbf{1} & \text{if }(w,\sigma)\models\phi\\
\mathbf{0} & \text{otherwise}
\end{array} ,\right.$ 
\item[-] $\left\Vert \varphi\oplus\psi\right\Vert _{\mathcal{V}}(w,\sigma)=\left\Vert \varphi\right\Vert _{\mathcal{V}}(w,\sigma) \sqcup \left\Vert
\psi\right\Vert _{\mathcal{V}}(w,\sigma),$
\item[-] $\left\Vert \varphi\otimes\psi\right\Vert _{\mathcal{V}}(w,\sigma) =\left\Vert \varphi\right\Vert _{\mathcal{V}}(w,\sigma
) \sqcap \left\Vert \psi\right\Vert _{\mathcal{V}}(w,\sigma),$
\item[-] $\left\Vert
{\textstyle\bigoplus\nolimits_{x}}
\centerdot\varphi\right\Vert _{\mathcal{V}}(w,\sigma)=\underset{0\leq i\leq
|w|-1}{\bigsqcup
}\left\Vert \varphi\right\Vert _{\mathcal{V}\cup\{x\}}(w,\sigma\lbrack
x\rightarrow i]),$ 
\item[-] $\left\Vert
{\textstyle\bigotimes\nolimits_{x}}
\centerdot\varphi\right\Vert _{\mathcal{V}}(w,\sigma)=\underset{0\leq i\leq
|w|-1}{\bigsqcap
}\left\Vert \varphi\right\Vert _{\mathcal{V}\cup\{x\}}(w,\sigma\lbrack
x\rightarrow i]),$ 
\item[-] $\left\Vert
{\textstyle\bigoplus\nolimits_{X}}
\centerdot\varphi\right\Vert _{\mathcal{V}}(w,\sigma)=\underset{I \subseteq dom(w)}{\bigsqcup
}\left\Vert \varphi\right\Vert _{\mathcal{V}\cup\{X\}}(w,\sigma\lbrack
X\rightarrow I]),$
\end{itemize}
\noindent where the operator $\underset{I \subseteq dom(w)}{\bigsqcup
}$ is applied on the ascending order according to the relation $\leq$ as defined in Remark \ref{order}. 
\end{definition}

We simply denote $\left\Vert \varphi\right\Vert _{\mathit{free}(\varphi)}$ by
$\left\Vert \varphi\right\Vert $, hence if $\varphi$ is a sentence, then $\left\Vert
\varphi\right\Vert \in K\left\langle \left\langle A^{\ast}\right\rangle
\right\rangle $.

\begin{lemma}\label{lemma-Dr} \cite{Dr:We}
Let $A$ be a linearly ordered alphabet, $\varphi \in MSO(K,A)$, and $\mathcal{V}$ be a finite set of variables containing $\mathit{free}(\varphi)$. Then    
$$  \left\Vert \varphi\right\Vert _{\mathcal{V}}(w,\sigma)
=  \left\Vert \varphi\right\Vert (w,\sigma|_{\mathit{free}(\varphi)})$$
for every $(w,\sigma)\in N_{\mathcal{V}}$. Furthermore $\left\Vert \varphi\right\Vert _{\mathcal{V}}$ is recognizable iff $\left\Vert \varphi\right\Vert $ is recognizable.
\end{lemma}
\begin{proof}
We extend the order on $A$ to a linear order on $A_{\mathcal{V}}$ and apply the proof of Prop. 3.3. in \cite{Dr:We} using our Theorems \ref{conj1}--\ref{inv-hom}. \hfill $\square$
\end{proof}

For first order variables $x,y,z$, second order variables $X_{1},\ldots,X_{m}$, and $\mathbf{k} \in K$ let 
\begin{itemize}
\item[] $\mathrm{first}(y):=\forall x\centerdot y\leq x,   \qquad \ \ \mathrm{last}(y):= \forall x \centerdot x\leq y,$
\item[] $(y=x+1):=\left(  (x\leq y)\wedge\lnot(y\leq x)\wedge\forall z\centerdot(z\leq x\vee y\leq z)\right),$
\item[] $\text{partition}(X_{1},\ldots,X_{m}):=\forall x\centerdot\underset
{i=1,\ldots,m}{\bigvee}\left(  (x\in X_{i})\wedge\underset{j\neq i}{\bigwedge
}\lnot(x\in X_{j})\right),$
\item[] $x \in X \rightarrow \mathbf{k}:= \lnot(x \in X) \oplus \left((x \in X) \otimes \mathbf{k} \right).$
\end{itemize}

Next we define a fragment of our MK-fuzzy MSO logic. 
\begin{definition} 
A formula $\varphi \in MKO(K,A)$ will be called \emph{restricted} if whenever it contains a subformula $\psi \otimes \psi'$, then $\psi$ is a (boolean) \emph{MSO} logic formula, and whenever it contains a subformula of the form ${\textstyle\bigotimes\nolimits_{x}}\centerdot\psi$, then $\psi$ is of the form $\underset{1 \leq i \leq m}{\textstyle\bigoplus} \left (   (x\in X_i) \rightarrow \mathbf{k_i}  \right)$, where $\mathbf{k_i} \in K$ for every $1 \leq i \leq m$.     
\end{definition}

We shall denote by $RMSO(K,A)$ the class of all restricted MK-fuzzy MSO logic formulas over $A$ and $K$. An MK-fuzzy language $s \in K\left\langle \left\langle A^{\ast}\right\rangle
\right\rangle $ is called RMSO-\emph{definable} if there is a sentence $\varphi \in RMSO(K,A)$ such that $s=\| \varphi \|$. 
The main result of this section is the subsequent theorem which follows from Theorems \ref{def-rec} and \ref{recdef} below.

\begin{theorem}
Let $A$ be a linearly ordered alphabet and $s \in K\left\langle \left\langle A^{\ast}\right\rangle
\right\rangle $. Then $s$ is recognizable iff it is \emph{RMSO}-definable.
\end{theorem}

\begin{theorem}\label{def-rec}
Let $A$ be a linearly ordered alphabet. If an MK-fuzzy language $s \in K\left\langle \left\langle A^{\ast}\right\rangle
\right\rangle $ is \emph{RMSO}-definable, then it is recognizable.
\end{theorem}
\begin{proof} [Sketch]
Let $\varphi \in RMSO(K,A)$ such that $s=\| \varphi \|$. We show by induction on the structure of $\varphi$ that $\|\varphi \| \in Rec(K,A)$. If  $\varphi =\mathbf{k}$ or $\varphi =\phi$,  then $\|\varphi\|$ is MK-fuzzy recognizable, respectively by Example \ref{constant} and Proposition \ref{rec-char_rec}. 
Next let $\varphi= \psi \oplus \psi'$ (resp. $\varphi= \psi \otimes \psi'$). We prove our claim using Lemma \ref{lemma-Dr} and Theorem \ref{Rec-disjunction} (resp. Lemma \ref{lemma-Dr} and Theorem \ref{conj1}). Assume now that $\varphi={\textstyle\bigoplus\nolimits_{x}}\centerdot\psi$ (resp. $\varphi={\textstyle\bigoplus\nolimits_{X}}\centerdot\psi$) such that $\|\psi\|$ is a recognizable MK-fuzzy language and let $\mathcal{V}=\mathit{free}(\varphi)$. We extend the order on $A_{\mathcal{V}}$ to a linear order on $A_{\mathcal{V}}\cup \{x\}$ (resp. $A_{\mathcal{V}}\cup \{X\}$) by letting $ (a,r[x=1]) \leq (a,r[x=0])$ (resp. $ (a,r[X=1]) \leq (a,r[X=0])$) for every  $(a,r) \in A_{\mathcal{V}} $. Then, we follow the proof of Lm. 4.3. in \cite{Dr:We} taking into account our Theorem \ref{hom} and show that $\|\varphi\|$ is recognizable. 
Finally, let $\varphi={\textstyle\bigotimes\nolimits_{x}}\centerdot \left(  \underset{1 \leq i \leq m}{\textstyle\bigoplus} \left (   (x\in X_i) \rightarrow \mathbf{k_i}  \right) \right) $ where $\mathbf{k_i} \in K$ for every $1 \leq i \leq m$. We consider the deterministic MK-fuzzy automaton $\mathcal{A}=(\{q\},q,T,  \{q\},in,wt,\ter)$ over $A_{\{X_1,\ldots,X_m \}}$ and $K$, with $T=\left \{(q, (a,r),q) \mid a\in A, r \in \{0,1\} ^{\{X_1,\ldots,X_m\}}\right \}$. The weight mappings are defined by $in(q)=\ter(q)=\mathbf{1}$ and $wt(q, (a,r),q)= \bigsqcup_{1 \leq i \leq m} \left(\mathbf{r}(X_i) \sqcap \mathbf{k_i}\right)$ for every $a\in A$ and $r \in \{0,1\} ^{\{X_1,\ldots,X_m\}}$, where $\mathbf{r}(X_i)=\mathbf{1}$ if $r(X_i)=1$ and $\mathbf{r}(X_i)=\mathbf{0}$ otherwise. Let $(w,\sigma) \in N_{\{X_1, \ldots, X_m \}}$, and assume that $(w,\sigma)=(a_0,r_0) \ldots (a_{n-1},r_{n-1})$ where $w=a_0 \ldots a_{n-1}\in A^*$ and $r_j \in \{0,1\} ^{\{X_1,\ldots,X_m\}}$ for every $0\leq j \leq n-1$. Then, there is a unique path $P_{(w,\sigma)}$ of $\mathcal{A}$ over $(w,\sigma)$. Moreover, we have 
\begin{align*}
\|\mathcal{A}\|(w,\sigma) & = weight(P_{(w,\sigma)})  = \bigsqcap_{0 \leq j \leq n-1} \left( \bigsqcup_{1 \leq i \leq m} \left(\mathbf{r_j}(X_i) \sqcap \mathbf{k_i}\right) \right )\\
&= \bigsqcap_{0 \leq j \leq n-1} \left( \left \|\left( \underset{1 \leq i \leq m}{\textstyle\bigoplus} \left (   (x\in X_i) \rightarrow \mathbf{k_i}  \right) \right)  \right \|_{\{x\}} (w, \sigma[x \rightarrow j])\right) \\
&= \left \|{\textstyle\bigotimes\nolimits_{x}}\centerdot \left(  \underset{1 \leq i \leq m}{\textstyle\bigoplus} \left (   (x\in X_i) \rightarrow \mathbf{k_i}  \right)  \right) \right \| (w, \sigma)   = \|\varphi\|(w,\sigma).
\end{align*}
Therefore, $\|\mathcal{A}\|=\|\varphi\|$, which implies that $\|\varphi\| \in Rec\left(K,A_{\{X_1, \ldots, X_m \}} \right)$, and this concludes our proof. \hfill $\square$
\end{proof}

For the converse of Theorem \ref{def-rec} we shall need the next lemma.

\begin{lemma}\label{in-ter-one}
Let $\mathcal{A}=(Q,I,T,F,in,wt,\ter)$ be an MK-fuzzy automaton over $A$ and $K$. Then there is an MK-fuzzy automaton $\mathcal{A'}=(Q',I',T',F',in',wt',\ter')$ over $A$ and $K$ such that $in'(q)=\mathbf{1}$ for every $q \in I'$ and $\ter'(q)=\mathbf{1}$ for every $q \in T'$, and $\|\mathcal{A'}\|(w)=\|\mathcal{A}\|(w)$ for every $w \in A^+$. 
\end{lemma}

\begin{theorem}\label{recdef}
Let $A$ be a linearly ordered alphabet. If an MK-fuzzy language $s \in K\left\langle \left\langle A^{\ast}\right\rangle
\right\rangle $ is recognizable, then it is \emph{RMSO}-definable.
\end{theorem}
\begin{proof} [Sketch]
Let $\mathcal{A} = (Q,I,T,F,in,wt,\ter)$ be an MK-fuzzy automaton over $A$ and $K$, and assume firstly that $\|\mathcal{A}\|(\varepsilon) = \mathbf{0}$. By Lemma \ref{in-ter-one}, we can assume that $in(q)=\mathbf{1}$ for every $q \in I$ and $\ter(q)=\mathbf{1}$ for every $q \in F$.  We intend to show that $\|\mathcal{A}\|$ is an RMSO-definable MK-fuzzy language. For this, we can follow the proof of Thm. 5.5. in \cite{Dr:We}. Nevertheless, in our case we have, in addition, to take care for the order of the paths of $\mathcal{A}$ over any word $w \in A^+$, as well as the order of the corresponding assignments. For every transition $(p,a,q)\in T$, we consider a second order 
variable $X_{p,a,q}$ and we let $\mathcal{V}=\{X_{p,a,q}\mid(p,a,q)\in T\}$. Let 
$m=\left\vert T\right\vert $. We define an enumeration $X_{1},\ldots,X_{m}$ of $\mathcal{V}$, preserving the order of the corresponding transitions in $T$.
We let \small
\begin{multline*}
\psi(X_{1},\ldots,X_{m})    :=\text{partition}(X_{1},\ldots,X_{m})\wedge\underset{(p,a,q)\in T}{\bigwedge}\forall x\centerdot\left(  \left(
x\in X_{p,a,q}\right)  \rightarrow P_{a}(x)\right)  \wedge \\
  \forall x\centerdot\forall y\centerdot\left(  \left(  y=x+1\right)
\rightarrow\underset{(p,a,q),(q,b,r)\in T}{\bigvee}\left(  x\in
X_{p,a,q}\right)  \wedge\left(  y\in X_{q,b,r}\right)  \right) \wedge \\
\exists z \centerdot \left( first(z) \wedge   \underset{p \in I } {\underset{(p,a,q) \in T } \bigvee} z \in X_{p,a,q} \right) \wedge\exists z' \centerdot \left( last(z') \wedge 
\underset{q \in F}{\underset{(p,a,q) \in T} \bigvee} z' \in X_{p,a,q} \right).
\end{multline*} \normalsize
Let $w=a_{0} \ldots a_{n-1}\in A^+$. We define a linear order on the set of all $(\mathcal{V},w)$-assignments satisfying $\psi$ in the following way. For two such assignments $\sigma$ and $\sigma'$, we let
$\sigma \leq \sigma'$ iff  there exists $k \in dom(w)$, with $0 \leq k \leq n-1$, such that $k \in \sigma(X_{i_k}) \cap \sigma'(X_{i'_k})$ with $i_k \leq i'_k$ and $j \in \sigma(X_{i_j}) \cap \sigma'(X_{i_j})$ for every $0 \leq j < k$. Trivially $\leq$ is a linear order. On the other hand, for every path $P_w$ of $\mathcal{A}$ over $w$ there exists a  unique $(\mathcal{V},w)$-assignment $\sigma_{P_w}$ satisfying $\psi$, i.e., $\left\Vert \psi\right\Vert (w,\sigma_{P_w})=\mathbf{1}$ and vice-versa (cf. Thm. 5.5. in \cite{Dr:We}). Then, we can easily get that $P_w \leq P'_w$ iff $\sigma_{P_w} \leq \sigma_{P'_w}$. 
Next, we consider the formula
$$
\varphi(X_{1},\ldots,X_{m}):=\psi(X_{1},\ldots,X_{m}) \otimes    \textstyle\bigotimes_x\centerdot 
\left( \underset{(p,a,q)\in T}{\textstyle\bigoplus}\left(  x\in
X_{p,a,q}\right)  \rightarrow wt(p,a,q)\right) .
$$
Let now $w=a_{0} \ldots a_{n-1}\in A^+$, $P_w  =\left (\left ( q_i, a_i,q_{i+1} \right ) \right )_{0 \leq i \leq n-1}$ a path of $\mathcal{A}$ over $w$, and $\sigma_{P_w}$ the corresponding ($\mathcal{V},w$)-assignment. Then, we get  $\left\Vert \varphi\right\Vert _{\mathcal{V}}(w,\sigma_{P_{w}})  = weight\left(P_w\right)$.
Finally, we consider the restricted MK-fuzzy MSO logic sentence 
$$\xi=\textstyle\bigoplus_{X_1}\ldots \textstyle\bigoplus_{X_m}\centerdot\varphi(X_{1},\ldots ,X_{m})$$ 
and we show that $\left\Vert \xi\right\Vert (w)  =  \Vert
\mathcal{A}\Vert(w)$ for every $w \in A^+$. 
Hence, $\| \mathcal{A}\| =\| \xi \|$, i.e., $\|\mathcal{A} \|$ is RMSO-definable. 
Next let $\|\mathcal{A}\|(\varepsilon) = \mathbf{k} \neq \mathbf{0}$. Then, by Lemma  \ref{in-ter-one},  we consider the MK-fuzzy automaton $\mathcal{A'}$ such that $\|\mathcal{A'}\|(w)=\|\mathcal{A}\|(w)$ for every $w \in A^+$. By what we have shown previously, there exists a restricted MK-fuzzy MSO logic sentence $\xi'$ such that $\|\mathcal{A'}\|=\|\xi\|$. We let  
$$\xi = \xi' \oplus \left( \forall x \centerdot \lnot (x \leq x) \otimes \mathbf{k}  \right). $$
Then $\xi$ is a restricted MK-fuzzy MSO logic sentence, and we get $\|\forall x \centerdot \lnot (x \leq x) \otimes \mathbf{k}\|(w)=\mathbf{0}$ for every $w \in A^+$, and $\|\forall x \centerdot \lnot (x \leq x) \otimes \mathbf{k}\|(\varepsilon)=\mathbf{k}$ (cf. \cite{Dr:We}). Hence $\|\mathcal{A}\|=\|\xi\|$, and this concludes our proof.  \hfill $\square$ 
\end{proof}

\

\section{Conclusion}
We introduced the bimonoid $K$ related to the fuzzification of MK-logic, and investigated MK-fuzzy automata over $K$. Our models are inspired by real practical applications being in development within the project LogicGuard \cite{Lo:G1,Lo:G2,Ku:Lo,DBLP:conf/rv/CernaSK16}. We proved properties of the class of MK-fuzzy languages accepted by MK-fuzzy automata as well as by their deterministic counterpart. We introduced an MK-fuzzy MSO logic and established a B\"uchi type theorem for the class of MK-fuzzy recognizable languages. 

It is worth noting that our results can be generalized to weighted automata over any bimonoid $(K,+,\cdot, 0,1)$ with the additional property that $0 \cdot k=0$ for every $k \in K$. Indeed, one can replace $\sqcup$ by $+$ and  $\sqcap$ by $\cdot$. 

Several problems remain open and they are under investigation, for instance, whether the class of recognizable MK-fuzzy languages is closed under MK-conjunction, Cauchy product and star operation, as well as whether the class of deterministically recognizable MK-fuzzy languages is closed under MK-disjunction and conjunction, Cauchy product, and star operation. Furthermore, due to the four-valued elements of $K$, there are several notions of supports and it is greatly desirable for applications to check which of them constitute recognizable languages. It should be clear from the proofs of our results, that the usual constructions on semiring-weighted automata cannot be always applied, even with modifications, when the weight structure is just a bimonoid. For instance, our bimonoid $K$ is zero-sum free and zero-divisor free. Nevertheless, one can not show that the support $\mathrm{supp}(s)=\{   w \in A^* \mid s(w) \neq \mathbf{0}  \}$ of a recognizable (even deterministically recognizable) MK-fuzzy language $s$ over $A$ and $K$ is a recognizable language following the usual construction on weighted automata (cf. for instance \cite{Dr:Au}). In our future research we intend also to study MK-fuzzy automata models over infinite words.

\nocite{*}
\bibliographystyle{eptcs}
\bibliography{MK_fuzzy_aut_GandALF_2017_final_version}

\end{document}